\documentclass[11pt] {article}

\usepackage[utf8]{inputenc}

\usepackage{amsmath,amsfonts,amsthm,amssymb,color}
\usepackage[usenames,dvipsnames,svgnames,table]{xcolor}
\definecolor{darkgreen}{rgb}{0.0,0,0.9}
\usepackage{mathtools}
\usepackage{authblk}
\usepackage{fullpage}
\usepackage{parskip}
\usepackage{comment}
\usepackage{tikz}
\usepackage{bbm}
\usepackage{dsfont}
\usepackage[sc]{mathpazo}
\usepackage[basic]{complexity}
\usepackage{algorithm2e}
\usepackage[colorlinks=true,
citecolor=OliveGreen,linkcolor=BrickRed,urlcolor=BrickRed,pdfstartview=FitH]{hyperref}
\usepackage[capitalize,nameinlink]{cleveref}
\usepackage{tcolorbox}
\usepackage[short]{optidef}

\newtcolorbox{wbox}
{
	colback  = white,
}

\SetKwInOut{Input}{Input}
\SetKwInOut{Output}{Output}
\SetKwFunction{Uncross}{\textsc{Uncross}}
\SetKwFunction{MergeUncross}{\textsc{MergeUncross}}
\SetKwFunction{PerfectMatching}{\textsc{PerfectMatching}}
\SetKwBlock{InParallel}{in parallel do}{end}
\SetKwFor{ParallelFor}{for}{in parallel do}{end}

\newcommand*{\suppress}[1]{}

\makeatletter
\def\thm@space@setup{%
	\thm@preskip= 10pt
	\thm@postskip=\thm@preskip 
}
\makeatother

\makeatletter
\renewcommand{\paragraph}{%
	\@startsection{paragraph}{4}%
	{\z@}{5pt}{-1em}%
	{\normalfont\normalsize\bfseries}%
}
\makeatother

\newtheorem{theorem}{Theorem}

\newtheorem{corollary}{Corollary}

\theoremstyle{definition}

\newenvironment{fminipage}%
{\begin{Sbox}\begin{minipage}}%
		{\end{minipage}\end{Sbox}\fbox{\TheSbox}}

\RequirePackage[numbers]{natbib}
\title{Two-Sided Matching Markets: Impossibility Results on \\
Existence of Efficient and Envy Free Solutions\footnote{Supported in part by NSF Grant CCF-2230414.}}

\date{\today}
\author{Thorben Tr\"obst \and Vijay V.\ Vazirani}

\begin{document}

\maketitle

\abstract{

The Hylland-Zeckhauser gave a classic pricing-based mechanism (HZ) for a one-sided matching market; it yields allocations satisfying Pareto optimality and envy-freeness \citep{hylland}, and the mechanism is incentive compatible in the large \citep{He2018pseudo}. They also studied the exchange extension of HZ and gave an example showing that it may not even admit an equilibrium. 

In this paper, we consider two models of two sided matching markets: when utility functions are symmetric and when they are non-symmetric. We ask if these models always admit allocations satisfying the two basic properties of Pareto efficiency and envy freeness. Our results are negative. A corollary of the former result is a negative result for non-bipartite matching markets as well. 
}

\bigskip
\bigskip
\bigskip
\bigskip
\bigskip
\bigskip
\bigskip
\bigskip
\bigskip
\bigskip
\bigskip
\bigskip
\bigskip
\bigskip
\bigskip
\bigskip
\bigskip
\bigskip
\bigskip
\bigskip
\bigskip
\bigskip
\bigskip
\bigskip

\pagebreak

\section{Introduction}
\label{sec.intro}

The classic Hylland-Zeckhauser (HZ) mechanism \cite{hylland} for a one-sided matching market uses the power of the pricing mechanism to find allocations satisfying several desirable properties, including Pareto optimality and envy-freeness \citep{hylland}, and the mechanism is incentive compatible in the large \citep{He2018pseudo}. Its drawback is that the problem of computing an HZ equilibrium is intractable: \cite{VY-HZ} showed that computing an $\epsilon$-approximate equilibrium is in the class PPAD and \cite{Chen-Y} showed the corresponding hardness result. 

In a sense, HZ is the most elementary matching market model under cardinal utilities; its status within matching markets is analogous to that of linear Fisher markets in General Equilibrium Theory. In the latter theory, the next level of generality is achieved by the linear Arrow-Debreu model, also known as the linear exchange model. In this model, agents have initial endowments of goods and they trade them, at equilibrium prices, in order to obtain bundles having higher utilities. Hylland and Zeckhauser studied the exchange extension of HZ as well, but gave an example showing that it may not even admit an equilibrium. 

That essentially put an end to the exploration of equilibria in matching market models; the only exception we are aware of is \cite{Echenique2019constrained}, which established a proof of existence of equilibrium in a hybrid model made up of a convex combination of an HZ and an exchange extension of HZ, provided there is a non-zero amount of the former. 

In this paper, we rekindle the original line of exploration by studying the case of two sided matching markets. We consider two models --- when utility functions are symmetric and when they are non-symmetric --- and ask if these models always admit allocations satisfying the two basic properties of Pareto efficiency and envy freeness. 

Let us take the two sides in such models to be agents and jobs, with each entity having a cardinal utility function over the set of $n$ entities on the other side. Note that an allocation in either model is a fractional perfect matching in the complete graph on agents and jobs. Our results are negative: for both models, we give instances for which we prove that any allocation fails to be either Pareto efficiency or envy free. We contrast our negative results with known positive results for the case of dichotomous, i.e., $0/1$, utility functions.


\section{Our Results}
\label{sec.results}

\begin{theorem}\label{thm:asymmetric}
For two-sided matching markets under asymmetric utilities, a Pareto-optimal and envy-free allocation does not always exist, even for the case of dichotomous, i.e., $0/1$, utility functions. 
\end{theorem}

\begin{proof}
    Consider the instance shown in Figure~\ref{fig:asymmetric} and the Pareto-optimal fractional
    perfect matching $y$ depicted in that figure.
    Let $x$ be some allocation (i.e.\ fractional perfect matching) in this instance and assume that
    $x$ is envy-free.
    We will show that $y$ is strictly Pareto-better than $x$.

    First let us show that $x_{24} = \frac{1}{3}$.
    Note that we must clearly have $x_{24} \geq \frac{1}{3}$ as otherwise $x_{25} > \frac{1}{3}$ or
    $x_{26} > \frac{1}{3}$ and in those cases agent $4$ would envy agent $5$ or $6$ respectively.
    On the other hand, assume that $x_{24} = \frac{1}{3} + \epsilon$.
    Then $u_2(x) \leq \frac{2}{3} - \epsilon$.
    But then agent 2 envies either agent 1 or agent 3 since among these three, one must get at least
    $\frac{2}{3}$ of agents 5 and 6.
    Thus $x_{24} = \frac{1}{3}$ as claimed.

    Next we claim that $x_{14} = \frac{1}{3}$.
    Again, we clearly have $x_{14} \geq \frac{1}{3}$ as otherwise agent 1 would envy agent 2 or
    agent 3.
    But in the other direciton, if $x_{14} = \frac{1}{3} + \epsilon$, then $x_{15} + x_{16} =
    \frac{2}{3} - \epsilon$.
    By the previous claim, we know that $x_{25} + x_{26} = \frac{2}{3}$ and so $x_{35} + x_{36} =
    \frac{2}{3} + \epsilon$ which would imply that agent 2 envies agent 3.
    Thus $x_{14} = \frac{1}{3}$.

    Finally, since $x_{24} = \frac{1}{3}$ and $x_{14} = \frac{1}{3}$, we can see that $y$ is
    Pareto-better than $x$ (regardless of how $x$ assigns the other edges).
    In particular, $u_1(y) = \frac{2}{3}$ whereas $u_1(x) = \frac{1}{3}$ and we have $u_i(y) \geq
    u_i(x)$ for all other $i$.
\end{proof}

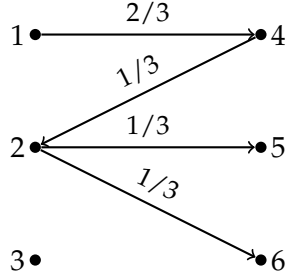
\begin{figure}[htb]
    \centering
    \begin{tikzpicture}[scale=1.5]
        \node[circle, fill, inner sep=1.5pt] (v1) at (0, 0) {};
        \node[circle, fill, inner sep=1.5pt] (v2) at (0, -1) {};
        \node[circle, fill, inner sep=1.5pt] (v3) at (0, -2) {};

        \node[circle, fill, inner sep=1.5pt] (w1) at (2, 0) {};
        \node[circle, fill, inner sep=1.5pt] (w2) at (2, -1) {};
        \node[circle, fill, inner sep=1.5pt] (w3) at (2, -2) {};

        \node[left] at (v1) {$1$};
        \node[left] at (v2) {$2$};
        \node[left] at (v3) {$3$};

        \node[right] at (w1) {$4$};
        \node[right] at (w2) {$5$};
        \node[right] at (w3) {$6$};

        \draw[->, thick] (v1) -- node[midway, above, sloped] {\small $2/3$} (w1);
        \draw[->, thick] (w1) -- node[midway, above, sloped] {\small $1/3$} (v2);
        \draw[->, thick] (v2) -- node[midway, above, sloped] {\small $1/3$} (w2);
        \draw[->, thick] (v2) -- node[midway, above, sloped] {\small $1/3$} (w3);
    \end{tikzpicture}
    \caption{Shown is a counterexample for dichotomous asymmetric weights. Each arrow represents a
    utility 1-edge from one side. All other utilities are 0. The edge labels represent a
    Pareto-optimal solution $y$ (assume that $y$ is a fractional perfect matching by filling up with
    utility 0 edges).}
    \label{fig:asymmetric}
\end{figure}

\bigskip

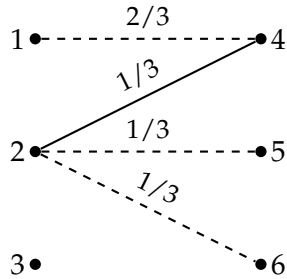
\begin{figure}[htb]
    \centering
    \begin{tikzpicture}[scale=1.5]
        \node[circle, fill, inner sep=1.5pt] (v1) at (0, 0) {};
        \node[circle, fill, inner sep=1.5pt] (v2) at (0, -1) {};
        \node[circle, fill, inner sep=1.5pt] (v3) at (0, -2) {};

        \node[circle, fill, inner sep=1.5pt] (w1) at (2, 0) {};
        \node[circle, fill, inner sep=1.5pt] (w2) at (2, -1) {};
        \node[circle, fill, inner sep=1.5pt] (w3) at (2, -2) {};

        \node[left] at (v1) {$1$};
        \node[left] at (v2) {$2$};
        \node[left] at (v3) {$3$};

        \node[right] at (w1) {$4$};
        \node[right] at (w2) {$5$};
        \node[right] at (w3) {$6$};

        \draw[-, dashed, thick] (v1) -- node[midway, above, sloped] {\small $2/3$} (w1);
        \draw[-, thick] (w1) -- node[midway, above, sloped] {\small $1/3$} (v2);
        \draw[-, dashed, thick] (v2) -- node[midway, above, sloped] {\small $1/3$} (w2);
        \draw[-, dashed, thick] (v2) -- node[midway, above, sloped] {\small $1/3$} (w3);
    \end{tikzpicture}
    \caption{Shown is a counterexample for symmetric weights. Dashed edges have utility 1, whereas
    solid edges have utility 2. All other utilities are 0. The edge labels represent a
    Pareto-optimal solution $y$ (assume that $y$ is a fractional perfect matching by filling up with
    utility 0 edges).}
    \label{fig:symmetric}
\end{figure}

\bigskip

\begin{theorem}\label{thm:symmetric}
For two-sided matching markets under symmetric utilities, a Pareto-optimal and envy-free allocation does not always exist, even if the utility functions are over $\{0, 1, 2\}$.
\end{theorem}

\begin{proof}
    Consider the instance shown in Figure~\ref{fig:symmetric} together with the depicted
    Pareto-optimal allocation $y$.
    Let $x$ be some envy-free allocation.
    We aim to show that $y$ is Pareto-better than $x$.

    First, we can once again see that $x_{24} = \frac{1}{3}$.
    Note that if $x_{24} < \frac{1}{3}$, then agent 4 will envy agent 5 or agent 6.
    Vice versa, if $x_{24} > \frac{1}{3}$, then agent 5 or 6 will envy agent 4.

    Next, note that $x_{14} = \frac{1}{3}$.
    Again, we must have $x_{14} \geq \frac{1}{3}$ since otherwise agent 1 would enyv agent 2 or
    agent 3.
    In the other direction, we cannot have $x_{14} > \frac{1}{3}$ since then agent 2 would envy
    agent 1 by the previous observation that $x_{24} = \frac{1}{3}$.

    Finally, we must have that $x_{25} = x_{26} = \frac{1}{3}$ since otherwise agent 5 would envy
    agent 6 or vice versa.
    This determines $x$ on all the edges with positive utility and thus all the utility values.
    But now we can see that $y$ is Pareto-better than $x$ since $u_1(y) > u_1(x)$ and $u_i(y) \geq
    u_i(x)$ for all other $i$.
\end{proof}

A non-bipartite matching market generalizes a two-sided matching markets under symmetric utilities. The graph of the latter inherently does not have odd cycles and therefore odd set constraints simply don't enter the picture. Hence we get:  

\begin{corollary}
	\label{cor.non-bip}
	For non-bipartite matching markets, a Pareto-optimal and envy-free allocation does not always exist, even if the utility functions are over $\{0, 1, 2\}$.
\end{corollary}

This raises the question of two-sided matching markets under symmetric dichotomous utilities and non-bipartite matching markets under dichotomous utilities. For the former, existence of equilibrium was shown in \cite{Bogomolnaia-2004-random} and a polynomial time algorithm for finding an equilibrium follows by applying the methods of \cite{VY-HZ}, who give such an algorithm for HZ under dichotomous utilities. For the latter, existence of equilibrium was shown in \cite{Roth-Kidney} and a polynomial time algorithm for finding an equilibrium was given in \cite{Li-2014-egalitarian}. Note that Theorem \ref{thm:asymmetric} also stands in contrast with the result of \cite{Bogomolnaia-2004-random}.

\bibliography{refs}

\end{document}